\documentclass[sip,biber]{now-journal} %
\usepackage{amsmath,amssymb,amsfonts}
\usepackage{graphicx,color}
\usepackage{textcomp}
\usepackage{xcolor}
\usepackage{algorithm}
\usepackage{algpseudocode}
\usepackage{algorithmicx}

\usepackage{subcaption}
\usepackage{amsthm}
\theoremstyle{definition}
\newtheorem{dfn}{Definition}
\newtheorem{prop}[dfn]{Proposition}
\usepackage{url}
\usepackage{xcolor}
\usepackage{adjustbox}
\usepackage{empheq}
\usepackage {booktabs}
\usepackage{adjustbox}
\usepackage{comment}
\usepackage{multirow}
\usepackage{makecell}
\newcommand{\bhline}[1]{\noalign{\hrule height #1}}

\title{Graph Signal Denoising Using Regularization by Denoising and Its Parameter Estimation}

\author[1*]{Kojima,Hayate}
\author[1]{Higashi,Hiroshi}
\author[1]{Tanaka,Yuichi}

\affil[1]{Graduate School of Engineering, The University of Osaka, Japan}

\issuevolumeyear{2025}
\issuevolumenumber{xx}
\issueelocation{e} %
\articletype{Original Paper} %

\articledatabox{Received xx xxxxx 2024; Revised xx xxxxx 2024\\[2pt]
ISSN 2048-7703; DOI 10.1561/116.xxxxxxxx\\
\copyright\ 2025 H. Kojima, H. Higashi and Y. Tanaka}

\OAstatement{This is an Open Access article, distributed under the terms of the Creative Commons Attribution licence (\url{http://creativecommons.org/licenses/by-nc/4.0/}), which permits unrestricted re-use, distribution, and reproduction in any medium, for non-commercial use, provided the original work is properly cited.}


\creditline*{Corresponding author: Hayate Kojima, h-kojima@msp-lab.org. The preliminary version of this paper is presented in \cite{kojima2024Interpretable}. This work is supported in part by JSPS KAKENHI under Grant 25KJ1728, 23H01415, 23K26110, and 23K17461.}

\begin{document}

\begin{abstract}
In this paper, we propose an interpretable denoising method for graph signals using regularization by denoising (RED). 
RED is a technique developed for image restoration that uses an efficient (and sometimes black-box) denoiser in the regularization term of the optimization problem. 
By using RED, optimization problems can be designed with the explicit use of the denoiser, and the gradient of the regularization term can be easily computed under mild conditions. 
We adapt RED for denoising of graph signals beyond image processing. 
We show that many graph signal denoisers, including graph neural networks, theoretically or practically satisfy the conditions for RED. 
We also study the effectiveness of RED from a graph filter perspective. 
Furthermore, we propose supervised and unsupervised parameter estimation methods based on deep algorithm unrolling. 
These methods aim to enhance the algorithm applicability, particularly in the unsupervised setting. 
Denoising experiments for synthetic and real-world datasets show that our proposed method improves signal denoising accuracy in mean squared error compared to existing graph signal denoising methods. 
\end{abstract}

\section{Introduction}
Graph signal processing (GSP) is a field of signal processing to analyze signals defined on graphs \cite{shumanEmergingFieldSignal2013,ortega2018Graph,tanaka2020Sampling}. 
By using GSP, we are able to efficiently analyze signals on networks, such as attributes on 3-D point clouds, sensor network data, and EEG, in which the relationship among data points is given by networks. 
GSP has attracted attention from various fields such as science, engineering, bioinformatics, and industry \cite{huang2018Graph,mutlu2012Signalprocessingbased}. 

Graph signals are often noisy due to measurement processes. 
Therefore, it is necessary to develop graph signal denoising techniques. 
Two approaches have mainly been proposed for graph signal denoising: One is model-based methods \cite{onuki2016Graph,pang2017Graph,chen2014Signal,ono2015Total}
and the other is data-driven ones \cite{kipf2017SemiSupervised,velickovic2018Graph,rey2022Untrained}. 

Model-based methods include filtering in the graph frequency domain \cite{onuki2016Graph} and (convex) optimization using total variation and/or graph Laplacian quadratic form \cite{pang2017Graph,chen2014Signal,ono2015Total}. 
While these methods are highly interpretable because of explicit formulations, they need to appropriately design an objective function, e.g., regularization, from the pre-determined signal prior. 

Data-driven methods typically use graph neural networks (GNNs) \cite{kipf2017SemiSupervised,velickovic2018Graph}. 
Graph convolution neural networks (GCNNs) \cite{kipf2017SemiSupervised} and graph attention networks (GATs) \cite{velickovic2018Graph} are representative GNN methods. 
They present promising performance for high-level tasks when they are sufficiently trained. 
In practice, however, such supervised learning requires a large amount of training data.
While several unsupervised learning \cite{rey2022Untrained} methods have been proposed, they require appropriate prior(s) to compensate for the limited training samples, similar to model-based approaches. 

To take a trade-off (or compromise) between these two approaches, combinations of model-based and data-driven methods have been proposed.
Typically, they start from a model-based objective function and learnable building blocks are included in its iterative optimization steps.
A widely used approach in this category for GSP is plug-and-play alternating direction method of multipliers (PnP-ADMM) \cite{chan2019Performance,yazaki2019Interpolation}, as an extension of PnP-ADMM in image processing to the graph setting.
It uses a (black-box) denoiser in its ADMM \cite{boyd2010Distributed} iteration.
While the obtained solution may be sub-optimal, its practical performance overcomes model- and data-driven approaches. 
However, due to the plugged-in denoiser, its overall formulation as an objective function is not always explicit which reduces the interpretability of the algorithm. 
It also requires many iterations to converge.

In this paper, we propose another method of a combination of model-based and data-driven approaches for graph signal denoising using regularization by denoising (RED) \cite{doi:10.1137/16M1102884}. 
RED has been utilized in image processing \cite{doi:10.1137/16M1102884, sun2021Dynamic, cohen2021Regularization, iskender2023REDPSM}.
In contrast to the PnP-ADMM, RED can explicitly include a black-box denoiser in its objective function (not in the internal algorithm) that makes the entire function interpretable. 
If the denoiser satisfies mild conditions, the gradient of the optimization problem can be represented in a simple form, which results in convergence to the global optimum solution. 
Note that it has not yet been studied whether high-performance graph signal denoisers, including GNNs, can be applicable for RED.

Our method is the first attempt to apply RED for GSP beyond image processing. 
We show that many graph signal denoising algorithms theoretically or practically satisfy the conditions for RED. 
Additionally, we propose both supervised and unsupervised parameter estimation methods based on deep algorithm unrolling \cite{monga2021Algorithm} and Noise2Noise \cite{lehtinen2018Noise2Noise}. 
We also reveal the effectiveness of RED from a graph filter perspective.

Experiments of graph signal denoising using synthetic and real-world datasets show that the proposed method overcomes existing model-based and data-driven methods as well as a combination of them.

The rest of the paper is organized as follows. 
Section \ref{sec:preliminaries} provides the preliminary of this paper, covering graph signal processing, RED, deep algorithm unrolling, and Noise2Noise. 
Section \ref{sec:graph-denoising} introduce popular graph signal denoising techniques. 
In Section \ref{sec:proposed}, we explain the proposed graph signal denoising method based on RED. 
Section \ref{sec:experiments} shows the efficacy of our proposed method through experiments on synthetic and real-world datasets. 
Finally, the conclusions of this paper are presented in Section \ref{sec:conclusion}.

\noindent\textit{Notation}:
Vectors are denoted by bold lowercase such as $\mathbf{x}$ and matrices by bold uppercase such as $\mathbf{X}$. The $i$-th component of the vector $\mathbf{x}$ is denoted by $x_i$ and the $(i,j)$ component of the matrix $\mathbf{X}$ is denoted by $X_{i,j}$. 
The element-wise product (Hadamard product) of two vectors is denoted by $\circ$.

\section{Preliminaries}
\label{sec:preliminaries}
In this section, we introduce fundamentals of GSP and RED. 
In addition, we explain parameter training methods used in this paper, deep algorithm unrolling \cite{monga2021Algorithm} and Noise2Noise \cite{lehtinen2018Noise2Noise}.  

\subsection{Graph Signal Processing}
A graph $\mathcal{G} = (\mathcal{V},\mathcal{E},\mathbf{W})$ is characterized by the node set $\mathcal{V} = \{v_1,v_2,\dots,v_N\}$, the edge set $\mathcal{E}$, and the weighted adjacency matrix $\mathbf{W} \in \mathbb{R}_{\geq 0}^{N \times N}$. 
The $(m,n)$ element of $\mathbf{W}$ is $W_{m,n}>0$ if $v_m$ and $v_n$ are connected by an edge. 
The number of nodes is given by $N=|\mathcal{V}|$. 
When the graph is undirected, graph Laplacian is defined by $\mathbf{L} = \mathbf{\Delta} - \mathbf{W}$ where $\mathbf{\Delta}$ is a diagonal matrix called a degree matrix whose elements are defined as $\Delta_{m,m} = \sum_n W_{m,n}$. 
Since $\mathbf{L}$ is a real symmetric matrix, it can be represented as $\mathbf{L} = \mathbf{U}\mathbf{\Lambda}\mathbf{U}^\top$ by eigendecomposition where $\mathbf{U} = [\mathbf{u}_1,\mathbf{u}_2,\dots,\mathbf{u}_N]\in \mathbb{R}^{N\times N}$ is an orthonormal matrix and $\mathbf{\Lambda} = \operatorname{diag}(\lambda_1,\lambda_2,\dots,\lambda_N)$ is a diagonal matrix of the corresponding eigenvalues. 

A graph signal $\mathbf{x}\in \mathbb{R}^N$ is a discrete signal defined on $\mathcal{V}$ (i.e., $x_n$ is located at node $v_n$). 
The graph Fourier transform (GFT) is defined as $\hat{\mathbf{x}} = \mathbf{U}^\top\mathbf{x}$ and inverse GFT (IGFT) is defined as $\mathbf{x} = \mathbf{U}\hat{\mathbf{x}}$.
The Laplacian quadratic form is defined by using the graph signal and the graph Laplacian as 
\begin{align}
    S_2(\mathbf{x}) &= \mathbf{x}^\top\mathbf{L}\mathbf{x}\\
    &=\sum_{(n,m)\in \mathcal{E}} W_{n,m}(x_n - x_m)^2,
    \label{eq:s2}
\end{align}
where $\mathcal{E}$ is a set of connected nodes. 
It is often used for a smoothness measure of graph signals.

Graph signals may contain noise like regular time-domain signals. The observed graph signal $\mathbf{y}$ can be modeled as
\begin{equation}
    \mathbf{y} = \mathbf{x}^* + \mathbf{n},
    \label{eq:observation_model}
\end{equation}
where $\mathbf{y} \in \mathbb{R}^N$ is the observed graph signal, $\mathbf{x}^*\in \mathbb{R}^N$ is the (unknown) original signal, $\mathbf{n}\in \mathbb{R}^N$ is additive white Gaussian noise.

In this paper, we consider a typical denoising problem: Obtaining a denoised signal $\mathbf{x}$ as an estimate of $\mathbf{x}^*$ from $\mathbf{y}$. 

\subsection{Regularization by Denoising (RED)}
\label{sec:RED}
RED \cite{doi:10.1137/16M1102884} is an image restoration method that uses image denoising algorithms for its regularizer.
Let us consider the observation model in \eqref{eq:observation_model} but $\mathbf{x}^*$ and $\mathbf{y}$ are now vectorized image signals.
RED considers the following optimization problem:
\begin{equation}
    \mathbf{x} = \underset{\tilde{\mathbf{x}}}{\operatorname{argmin}}~ \frac{1}{2} \|\tilde{\mathbf{x}} - \mathbf{y}\|^2_2 + \frac{\alpha_{\text{red}}}{2} \tilde{\mathbf{x}}^\top(\tilde{\mathbf{x}} - \mathcal{D}_{\text{image}}(\tilde{\mathbf{x}}))
    \label{eq:RED}
\end{equation}
where $\alpha_{\text{red}}$ is a regularization parameter and $\mathcal{D}_{\text{image}}(\cdot)$ is an image denoiser.

While this optimization problem is typically solved by iterative optimization algorithm such as gradient descent or ADMM, it is necessary to compute the gradient of \eqref{eq:RED}. 
In \cite{doi:10.1137/16M1102884}, it has been shown that the gradient of the regularization term can be rewritten as
\begin{equation}
    \nabla \left(\frac{1}{2}\tilde{\mathbf{x}}^\top(\tilde{\mathbf{x}} - \mathcal{D}_{\text{image}}(\tilde{\mathbf{x}}))\right) =  \tilde{\mathbf{x}} - \mathcal{D}_{\text{image}}(\tilde{\mathbf{x}})
    \label{eq:RED-grad}
\end{equation}
when $\mathcal{D}_{\text{image}}(\cdot)$ satisfies the following two conditions:
\begin{itemize}
    \item \textbf{(Local) Homogeneity:} $\mathcal{D}_{\text{image}}(c\cdot\tilde{\mathbf{x}}) = c\cdot\mathcal{D}_{\text{image}}(\tilde{\mathbf{x}})$ around $c = 1$,
    \item \textbf{Strong Passivity:}
    $\eta(\nabla \mathcal{D}_{\text{image}}(\tilde{\mathbf{x}}))\leq 1$,
\end{itemize}
where $\eta(\cdot)$ denotes the spectral radius.
In other words, if the above conditions are satisfied, the minimization problem can be solved using an iterative algorithm without computing the gradient of the inner denoiser. 
In gradient descent, for example, the update step of the algorithm is written as follows:
\begin{equation}
    \tilde{\mathbf{x}}^{(k+1)} = \tilde{\mathbf{x}}^{(k)} - \epsilon \left(\tilde{\mathbf{x}}^{(k)} - \mathbf{y} + \alpha_\text{red} \left(\tilde{\mathbf{x}}^{(k)} - \mathcal{D}_{\text{image}}(\tilde{\mathbf{x}}^{(k)})\right)\right),\label{eq:red-update}
\end{equation}
where $\epsilon$ is the step size. 
It has been shown that many high-performance \textit{image} denoising methods satisfy the above conditions (at least, practically).

\subsection{Deep Algorithm Unrolling}
Deep algorithm unrolling (DAU) is a family of methods for training parameters in iterative optimization algorithms using deep learning techniques \cite{li2020Efficient}. 
Roughly speaking, DAU \textit{unrolls} the iterative algorithm and the parameters in the unrolled iterations are learned using backpropagation from training data. 
Although the unrolled steps with the learned parameters may not be guaranteed to converge to the optimal solution in general, practical improvements have been reported in terms of convergence speed and accuracy compared to conventional methods \cite{li2020Efficient,shi2022Algorithm}.
For more details, please refer to \cite{monga2021Algorithm}.

\subsection{Noise2Noise}
Noise2Noise \cite{lehtinen2018Noise2Noise} is a deep learning-based image denoising model. It is designed to train parameters in neural networks under the unsupervised setting.

First, let us consider training for the supervised setting.
When using the mean squared error (MSE) as the loss function, parameters in a neural network are learned to minimize the following cost function.
\begin{equation}
    \boldsymbol{\theta}_{\text{N2C}} = \underset{\boldsymbol{\theta}}{\arg\min}~\mathbb{E}\left[\|f_{\boldsymbol{\theta}}(\mathbf{y}) - \mathbf{x}^*\|_2^2\right],
\end{equation}
where $\boldsymbol{\theta}$ is a parameter vector and $f_{\boldsymbol{\theta}}(\cdot)$ is the output under $\boldsymbol{\theta}$. 
This clearly requires $\mathbf{x}^*$: The ground-truth clean images.
However, this is often not the case.

In Noise2Noise, it makes pseudo pairs of images by intentionally adding noise to the observed noisy signals.
Hence, the to-be-denoised signal $\mathbf{y}_{\text{noisy}}$ is given by
\begin{equation}
    \mathbf{y}_{\text{noisy}} = \mathbf{y} + \mathbf{n}_{\text{noisy}},
    \label{eq:n2n}
\end{equation}
where $\mathbf{n}_{\text{noisy}}$ is additive white Gaussian noise conforming to $\mathcal{N}(0, \sigma_{\text{N2N}})$. 
Note that $\mathbf{y}$ is already noisy.

By using $\mathbf{y}$ and $\mathbf{y}_{\text{noisy}}$ and assuming $\mathbf{y} \sim \mathbf{x}$, \eqref{eq:n2n} is modified as follows:
\begin{equation}
    \boldsymbol{\theta}_{\text{N2N}} = \underset{\boldsymbol{\theta}}{\arg\min}~\mathbb{E}\left[\|f_{\boldsymbol{\theta}}(\mathbf{y}_{\text{noisy}}) - \mathbf{y}\|_2^2\right]. \label{eq:noise2noise}
\end{equation}
Finally, $f_{\boldsymbol{\boldsymbol{\theta}_{\text{N2N}}}}(\cdot)$ is used for denoising of $\mathbf{y}$.

\section{Graph Signal Denoising}
\label{sec:graph-denoising}
In this section, we introduce existing popular denoising algorithms for graph signals.

\subsection{Graph Laplacian Regularization}
\label{sec:LR}
Let us consider the observation model in \eqref{eq:observation_model}.
Graph Laplacian regularization \cite{pang2017Graph} is a well-known model-based graph signal denoising method. It is based on the following minimization problem:
\begin{equation}
    \mathbf{x} = \underset{\tilde{\mathbf{x}}}{\operatorname{argmin}}~ \frac{1}{2}\|\tilde{\mathbf{x}}-\mathbf{y}\|^2_2 + \frac{\alpha_{\text{lr}}}{2}S_2(\tilde{\mathbf{x}})
    \label{eq:lr}
\end{equation}
where $\alpha_{\text{lr}}$ is a regularization parameter.
The second term corresponds to regularization using signal smoothness with the Laplacian quadratic form shown in \eqref{eq:s2}. 
Its solution is given by the following closed-form:
\begin{align}
    \mathbf{x} &= (\mathbf{I} + \alpha_\text{lr}\mathbf{L})^{-1}\mathbf{y} := \mathcal{D}_{\text{LR}}(\mathbf{y}).
    \label{eq:lr-closed}
\end{align}

Since the graph Laplacian often becomes a large matrix, the computational cost for the inverse matrix in \eqref{eq:lr-closed} is also large. 
To reduce the computational cost, approximate solutions of the inverse can be obtained by polynomial approximation, conjugate gradient, and ADMM \cite{boyd2010Distributed}.

\subsection{PnP-ADMM}
Graph signal denoising using PnP-ADMM \cite{yazaki2019Interpolation} is based on the following optimization problem:
\begin{align}
    \underset{\tilde{\mathbf{x}},\mathbf{v}}{\operatorname{min}}~ \frac{1}{2}\|\tilde{\mathbf{x}}-\mathbf{y}\|^2_2 + \alpha_{\text{pnp}} s(\mathbf{\mathbf{v}}) 
    ~~\text{s.t.}~ \tilde{\mathbf{x}} = \mathbf{v} \label{eq:pnp-admm}
\end{align}
where $\alpha_{\text{pnp}}$ is a regularization parameter and $s(\cdot)$ is an implicit regularization function.
When ignoring the convexity of \eqref{eq:pnp-admm}, the following ADMM iteration can be applied:
\begin{empheq}[left={\empheqlfloor}]{align}
  \mathbf{x}^{(k+1)} &= \frac{1}{1+\rho}(\mathbf{x}^{(k)} + \rho(\mathbf{v}-\mathbf{u}))\nonumber\\
  \mathbf{v}^{(k+1)} &= \mathcal{D}_{\text{graph}}(\mathbf{x}^{(k+1)}+\mathbf{u}^{(k)}) \label{eq:admm-update} \\
  \mathbf{u}^{(k+1)} &= \mathbf{u}^{(k)} + (\mathbf{x}^{(k+1)} - \mathbf{v}^{(k+1)}), \nonumber
\end{empheq}
where $k$ is the iteration number, $\mathcal{D}_{\text{graph}}(\cdot)$ is a black-box graph signal denoiser implicitly related to $s(\mathbf{\mathbf{v}})$, $\rho$ is the stepsize, and $\mathbf{v}$ and $\mathbf{u}$ are auxiliary variables. 

In practice, PnP-ADMM outperforms the model- and data-driven approaches for graph signal denoising \cite{yazaki2019Interpolation,cai2025Unrolling}.
However, its overall optimization problem is not always explicit since $s(\cdot)$ cannot be derived from $\mathcal{D}_{\text{graph}}(\cdot)$ straightforwardly.

\section{RED-based Graph Signal Denoising}
\label{sec:proposed}
In this section, we introduce our proposed method for denoising graph signals using RED.
The denoised graph signal is obtained by replacing the image denoiser $\mathcal{D}_\text{image}$ in \eqref{eq:RED} with a graph signal denoiser $\mathcal{D}_\text{graph}$ as
\begin{equation}
    \mathbf{x} = \underset{\tilde{\mathbf{x}}}{\operatorname{argmin}}~ \frac{1}{2} \|\tilde{\mathbf{x}} - \mathbf{y}\|^2_2 + \frac{\alpha_{\text{red}}}{2} \tilde{\mathbf{x}}^\top(\tilde{\mathbf{x}} - \mathcal{D}_{\text{graph}}(\tilde{\mathbf{x}})).
    \label{eq:proposed}
\end{equation}
While the objective function is similar to \eqref{eq:RED} except for that $\mathbf{x}$ and $\mathbf{y}$ are graph signals and $\mathcal{D}_{\text{image}}$ is replaced with a graph signal denoiser $\mathcal{D}_{\text{graph}}$, the applicability of RED to the graph setting is not trivial.

In this section, we first show that widely-used graph signal denoisers are applicable to RED. 
Then, we propose supervised and unsupervised parameter training methods using DAU and Noise2Noise. 
Finally, the effectiveness of RED is revealed from a graph filter perspective.

\subsection{Are Graph Signal Denoisers Applicable to RED?}
\label{sec:RED2GSP}
For the regularizer of RED, the two conditions shown in Section \ref{sec:RED} must be satisfied.
Here, we show that the graph Laplacian regularization (LR) \cite{pang2017Graph}, GAT \cite{velickovic2018Graph}, and even PnP-ADMM \cite{yazaki2019Interpolation} satisfy these conditions.

\subsubsection{LR}
\noindent\textbf{(Local) Homogeneity}:
The solution with LR in \eqref{eq:lr-closed} can be viewed as a graph lowpass filtering by $h(\mathbf{L}) = (\mathbf{I} + \alpha_{\text{lr}}\mathbf{L})^{-1}$ applying to the observed signal $\mathbf{y}$. 
If $\mathbf{L}$ is determined independently of the graph signal, $h(\mathbf{L})(c\cdot\mathbf{x}) = c\cdot h(\mathbf{L})\mathbf{x}$ is always satisfied.

We consider another cases that the underlying graph is estimated directly from the graph signal.
In this paper, we focus on the following two representative situations: 
\begin{itemize}
    \item The edge weights are inversely proportional to the distance between signal values.
    \item The graph Laplacian is estimated using a Gaussian Markov Random Field (GMRF).
\end{itemize}
Below, proofs of homogeneity for these cases are shown.

\begin{prop}
    Suppose that the adjacency matrix $\mathbf{W}$ is constructed with edge weights $W_{i,j}$ for any connected nodes $v_i$ and $v_j$ as:
    \begin{equation}
        W_{i,j} = 1/\sqrt{\| y_i - y_j \|^2}.
        \label{eq:calc-W}
    \end{equation}
    If $\mathbf{W}$ is normalized, the Laplacian regularization denoiser $\mathcal{D}_{\text{LR}}(\cdot)$ defined in \eqref{eq:lr-closed} satisfies $\mathcal{D}_\text{LR}(c\cdot\mathbf{y}) = c \cdot \mathcal{D}_\text{LR}(\mathbf{y})$.
\end{prop}

\begin{proof}
    Let us denote the scaled version of $\mathbf{W}$ as $\mathbf{W}'$.
    The element in $\mathbf{W}'$ can then be written as
    \begin{equation}
        W'_{i,j} = 1/ \left(c \sqrt{\|y_i - y_j\|^2}\right).
        \label{eq:c-multiplied-W}
    \end{equation}
    This immediately results in local homogeneity.
\end{proof}

\begin{prop}
    Suppose that the original signal $\mathbf{x}^*$ is modeled as a Gaussian Markov random field (GMRF), whose probability density function (PDF) $p(\cdot)$ is given by \cite{lawrence2012Unifying}:
    \begin{align}
    p(\mathbf{y}|\boldsymbol{\mu},\mathbf{Q}) &= \frac{(\operatorname{det}\mathbf{Q})^{1/2}}{(2\pi)^{N/2}}\exp\left(-\frac{1}{2}(\mathbf{y}-\boldsymbol{\mu})^\top\mathbf{Q}(\mathbf{y}-\boldsymbol{\mu})\right), \label{eq:gmrf-a}\\
    \mathbf{Q} &= \mathbf{L}+\delta\mathbf{I},\label{eq:gmrf-q}
    \end{align}
    where $\boldsymbol{\mu}$ is the mean vector of $\mathbf{y}$, $\mathbf{Q}$ is the inverse covariance matrix, i.e., precision matrix, and $\delta >0$ is a parameter. 
    If $\mathbf{L}$ is normalized, the Laplacian regularization denoiser $\mathcal{D}_{\text{LR}}(\cdot)$ in this case satisfies $\mathcal{D}_\text{LR}(c\cdot\mathbf{y}) = c \cdot \mathcal{D}_\text{LR}(\mathbf{y})$.
\end{prop}
\begin{proof}
    We examine how the graph Laplacian $\mathbf{L}$ changes when the signal is scaled. This is done by comparing the GMRF model for a signal $\mathbf{y}$ with the model for its scaled version, $\mathbf{y}' = c\mathbf{y}$. 
    The probability density function for $\mathbf{y}'$ can be written using the change of variables formula as \cite{bishop2006Pattern}:
    \begin{align}
        p(\mathbf{y}'|\boldsymbol{\mu}',\mathbf{Q}') = p(\mathbf{y}|\boldsymbol{\mu},\mathbf{Q})\left|\operatorname{det}\left(\frac{d\mathbf{y}}{d\mathbf{y}'}\right)\right|\label{eq:gmrf-b}
    \end{align}
    where the Jacobian determinant is given by
    \begin{align}
        \left|\operatorname{det}\left(\frac{d\mathbf{y}}{d\mathbf{y}'}\right)\right| &= \left|\operatorname{det}(c^{-1}\mathbf{I})\right| \\
        &= (1/c)^N.
    \end{align}
    Thus, \eqref{eq:gmrf-b} can be rewritten as
    \begin{equation}
    \begin{split}
        &p(\mathbf{y}'|\boldsymbol{\mu}',\mathbf{Q}') = p(\mathbf{y}|\boldsymbol{\mu},\mathbf{Q})\left|\operatorname{det}\left(\frac{d\mathbf{y}}{d\mathbf{y}'}\right)\right|\\
        &= p(\mathbf{y}'/c|\boldsymbol{\mu},\mathbf{Q})(1/c)^{N}\\
        &= \frac{(\operatorname{det}\mathbf{Q})^{1/2}}{(2\pi)^{N/2}}\exp\left(-\frac{1}{2}(\mathbf{y}'/c-\boldsymbol{\mu})^\top\mathbf{Q}(\mathbf{y}'/c-\boldsymbol{\mu})\right)(1/c)^{N}\\
        &= \frac{(\operatorname{det}(\mathbf{Q}/c^2))^{1/2}}{(2\pi)^{N/2}}\exp\left(-\frac{1}{2}(\mathbf{y}' -c\boldsymbol{\mu})^\top\frac{\mathbf{Q}}{c^2}(\mathbf{y}' - c\boldsymbol{\mu})\right).
        \label{eq:gmrf-c}
    \end{split}
    \end{equation}
    If $\mathbf{L}$ is normalized, the term $1/c^2$ for $\mathbf{Q}$ can be ignored from \eqref{eq:gmrf-q} and the scaling term only affects the mean of the signal: This results in homogeneity.
    
\end{proof}

\noindent\textbf{Strong Passivity}:
According to \cite{doi:10.1137/16M1102884}, when (local) homogeneity is satisfied, it is known that $\nabla\mathcal{D}_{\text{graph}}(\mathbf{x})\mathbf{x} = f(\mathbf{L})\mathbf{x}$.
Therefore, the condition is satisfied when $\eta(f(\mathbf{L}))\leq 1$. 
Since $\mathbf{L}$ is a positive-semidefinite matrix, the maximum eigenvalue of $(\mathbf{I} + \alpha\mathbf{L})^{-1}$ is always less than or equal to 1.
This implies that strong passivity is satisfied.

\subsubsection{GAT and PnP-ADMM} \label{subsubsec:gat-and-pnp}
Both GAT and PnP-ADMM are black-box denoisers (at least, partially).
Therefore, we experimentally validate that they satisfy the two conditions for RED using the synthetic and real-world 3-D point cloud datasets employed in our experiments (see Section \ref{sec:experiments}-\ref{sec:experiments-setup} for further details).

\noindent\textbf{(Local) Homogeneity}:
Fig. \ref{fig:cond1} shows a comparison of $\mathcal{D}_{\text{graph}}(c\cdot \mathbf{y})$ and $c\cdot \mathcal{D}_{\text{graph}}(\mathbf{y})$, and it is clear that GAT and PnP-ADMM have (local) homogeneity for a small $c$. 

\noindent\textbf{Strong Passivity}:
Fig. \ref{fig:cond2} shows $\|\mathcal{D}_{\text{graph}}(\mathbf{y})\|^2/\|\mathbf{y}\|^2$ in various datasets. 
Both GAT and PnP-ADMM show $\|\mathcal{D}_{\text{graph}}(\mathbf{y})\|^2 / \|\mathbf{y}\|^2 \leq 1$ for all datasets.
This implies the denoisers are bounded: They can be practically used for RED.
The above facts imply that many graph signal denoisers $\mathcal{D}_{\text{graph}}$ can be applied to RED for graph signals. 

Note that, a well-trained GNN for signal denoising can be readily expected to satisfy these conditions. 
Furthermore, it is also possible to use GNNs that are theoretically proven to satisfy them \cite{velasco2024Graph, nt2019Revisiting, jia2023Stabilizing}. 

We should be careful when using methods where dynamic ranges between the input and output change significantly, such as GNNs for graph classification.
In this case, it is necessary to verify that the conditions are satisfied. 

\begin{figure*}[tp]
  \centering
  \begin{minipage}[]{\linewidth}
    \centering
    \includegraphics[height=4.5cm]{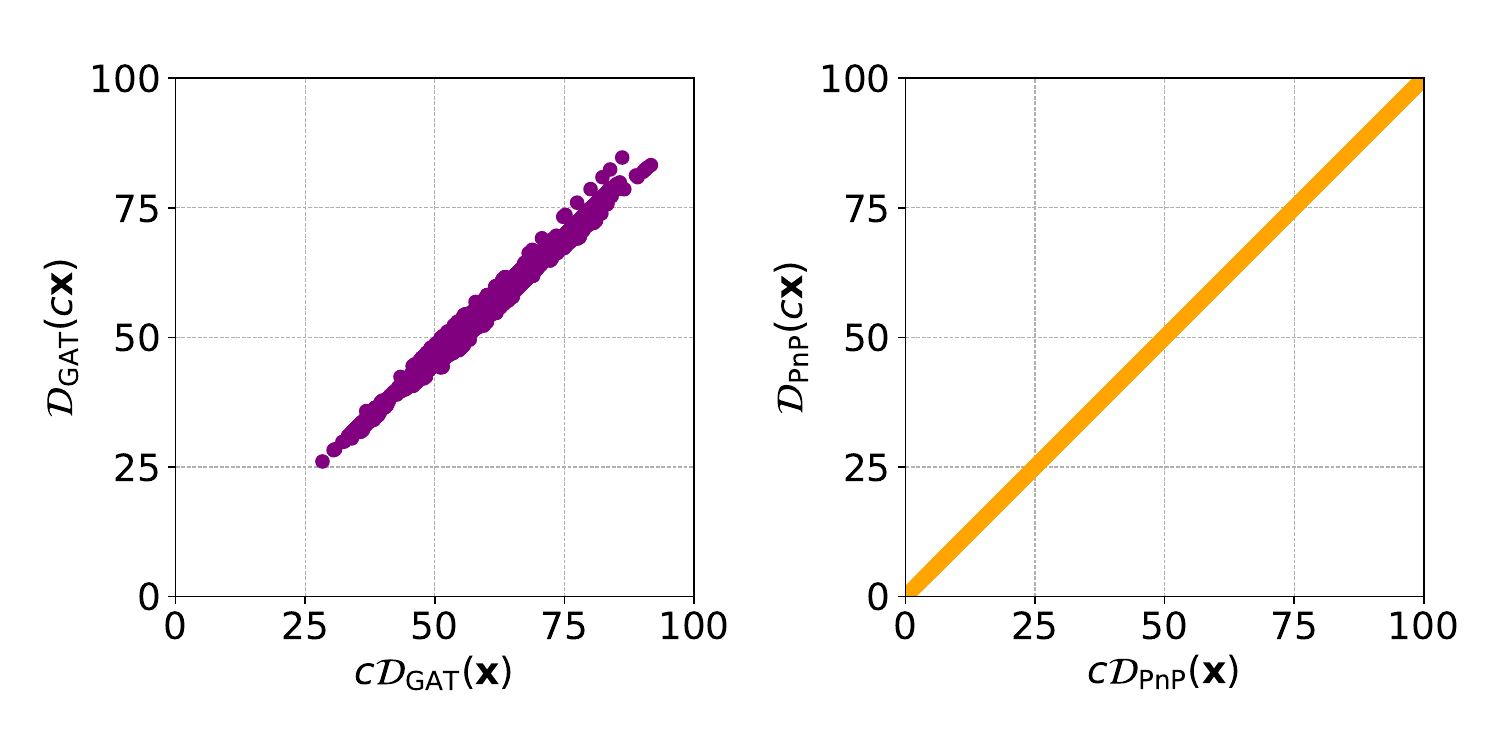}
    \subcaption{}
    \label{fig:cond1}
  \end{minipage}
  \begin{minipage}[]{\linewidth}
    \centering
    \includegraphics[height=4cm]{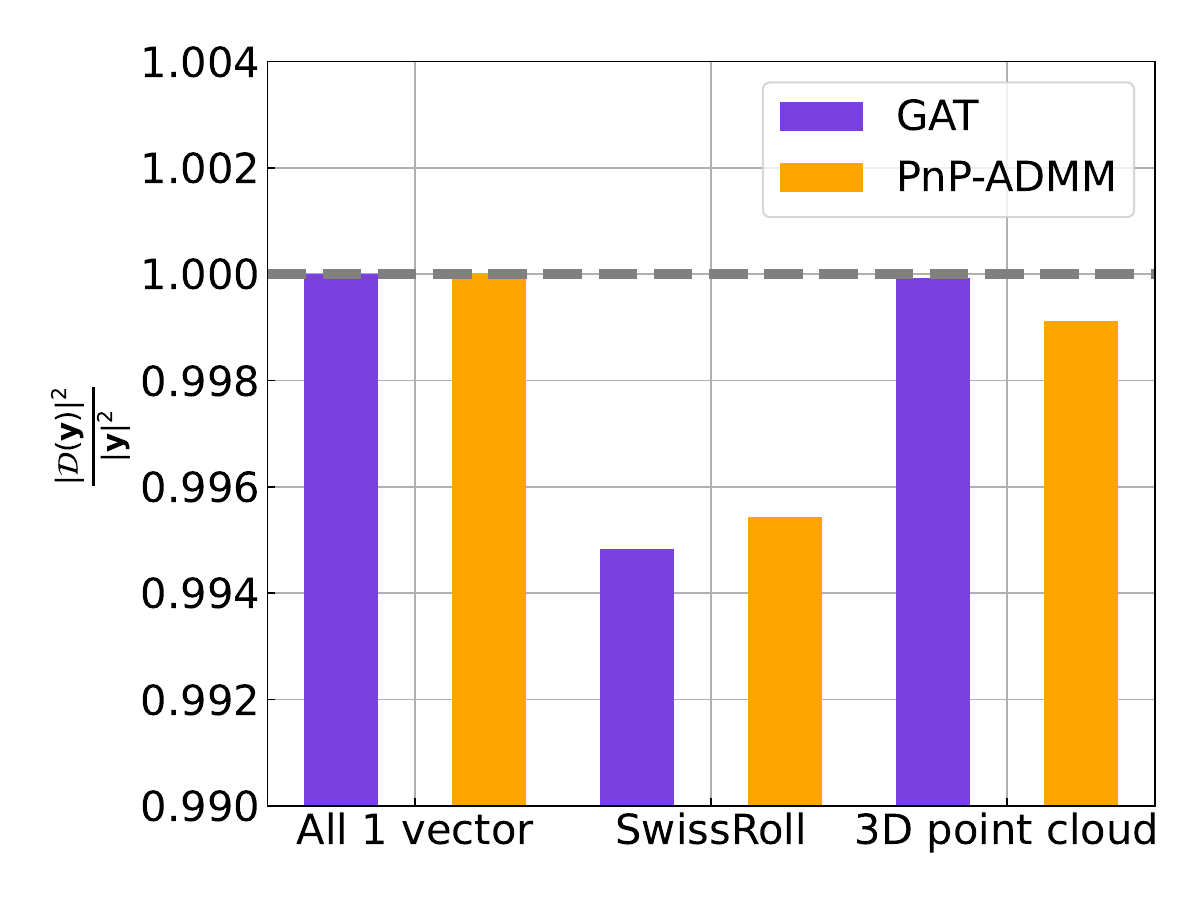}
    \subcaption{}
    \label{fig:cond2}
  \end{minipage}
  \caption{Applicability of GAT and PnP-ADMM for RED. (a) Scatter plot of $\mathcal{D}_{\text{graph}}(c\mathbf{x})$ vs $c\mathcal{D}_{\text{graph}}(\mathbf{x})$ ($c=1.1$). Left: GAT. Right: PnP-ADMM. We use a real-world point cloud dataset (details are shown in Section \ref{sec:experiments-setup}). (b) Bar plot of $\|\mathcal{D}_{\text{graph}}(\mathbf{y})\|^2/\|\mathbf{y}\|^2$. ``All 1 vector'' refers to $[1, \dots, 1]^\top$.}
\end{figure*}

\begin{figure}[tbp]
    \centering
    \includegraphics[width=4.5cm]{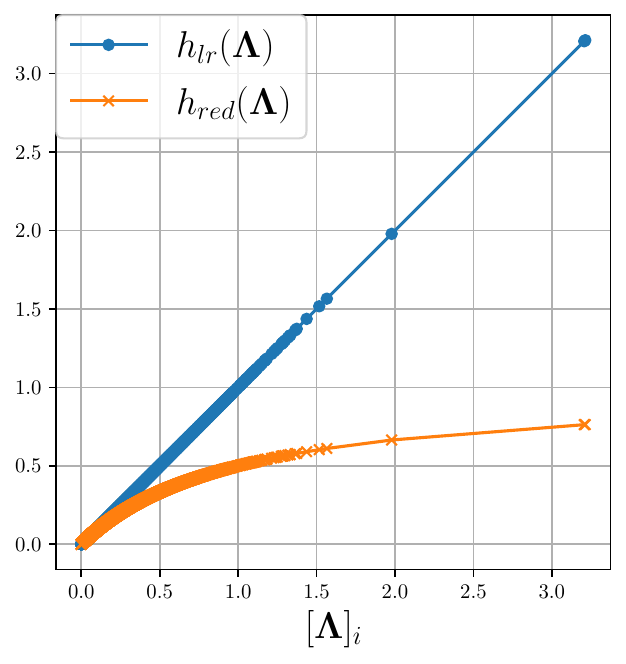}
    \caption{Comparion of $h_\text{lr}(\boldsymbol{\Lambda})$ and $h_\text{red}(\boldsymbol{\Lambda})$ using 3-D point cloud dataset. $\alpha_\text{lr}$ and $\alpha_\text{red}$ are fitted as same as in experiments.}
    \label{fig:LRvsRED}
\end{figure}

\vspace{0.5cm}
Based on the above empirical results, the RED algorithm can be applied to graph signal denoising in practice. 
We only need to replace $\mathcal{D}_\text{image}$ in \eqref{eq:RED} with $\mathcal{D}_\text{graph}$.
\eqref{eq:proposed} can be solved by one of many iterative algorithms such as conjugate gradient (shown in Algorithm \ref{algorithm:proposed-red}) and ADMM \cite{boyd2010Distributed}. 

Note that both LR and PnP internally require matrix inversion or eigendecomposition, making them computationally expensive for large-scale graphs. 
To address this issue, several approximation techniques could be straightforwardly employed. 
For instance, the matrix inversion can be efficiently approximated using methods like Chebyshev polynomial approximation \cite{hammond2011Wavelets,shuman2011Chebyshev}. 
Furthermore, it is also possible to incorporate existing denoising methods designed for large-scale graphs \cite{zhang2023Subgraph,wu2023SGFormer}, such as those based on distributed processing or graph partitioning, as the inner denoiser within our framework.

\begin{algorithm}[tb]
    \caption{Graph Signal Denoising Using RED}
    \label{algorithm:proposed-red}
    \begin{tabular}[t]{@{}ll} %
        \textbf{Input:} \\
        $\mathbf{Y}$ & : Observed signal \\
        $K$ & : Number of layers \\
        $\alpha_\text{RED}$ & : Hyperparameter in \eqref{eq:proposed} \\
        $\alpha_\text{denoiser}$ & : Hyperparameter used in $\mathcal{D}_\text{graph}$ \\
        \textbf{Output:} \\
        $\tilde{\mathbf{x}}^{(K)}$ & : Denoised signal
    \end{tabular}
    
    \hrulefill
    \begin{algorithmic}[1]
        \State{\textbf{Initialization: }}
        \State{~~~$\tilde{\mathbf{x}}^{(0)} = \mathbf{0}$}
        \State{~~~$\mathbf{x}_\nabla^{(0)} = \tilde{\mathbf{x}}^{(0)} - \mathbf{y} + \alpha_\text{red} \left(\tilde{\mathbf{x}}^{(0)} - \mathcal{D}_{\text{graph}}(\tilde{\mathbf{x}}^{(0)})\right)$}
        \State{~~~$\mathbf{x}_\Delta = - \mathbf{x}_\nabla^{(0)}$}

        \For{$k = 1$ to $K$}
            \State{$\tau = - \frac{\sum \mathbf{x}_\Delta \circ \mathbf{x}_\nabla^{(k-1)}}{\sum \mathbf{x}_\Delta \circ \left( \mathbf{x}_\Delta + \alpha_\text{red} \left(\mathbf{x}_\Delta - \mathcal{D}_{\text{graph}}(\mathbf{x}_\Delta)\right)\right)}$}
            \State{$\tilde{\mathbf{x}}^{(k)} = \tilde{\mathbf{x}}^{(k-1)} + \tau \mathbf{x}_\Delta$}
            \State{$\mathbf{x}_\nabla^{(k)} = \tilde{\mathbf{x}}^{(k)} - \mathbf{y} + \alpha_\text{red} \left(\tilde{\mathbf{x}}^{(k)} - \mathcal{D}_{\text{graph}}(\tilde{\mathbf{x}}^{(k)})\right)$}
            
            \State{$\gamma = \frac{\| \mathbf{x}^{(k)}_\nabla \|_F^2 }{\| \mathbf{x}^{(k-1)}_\nabla \|_F^2 }$}
            \State{$\mathbf{x}_\Delta = - \mathbf{x}^{(k)}_\nabla + \gamma \mathbf{x}_\Delta$}
        \EndFor
    \end{algorithmic}
\end{algorithm}

\subsection{Parameter Learning Using Deep Algorithm Unrolling and Noise2Noise}
Our proposed algorithm is shown in Algorithm \ref{algorithm:proposed-red}.
In fact, \eqref{eq:proposed} requires estimating the regularization parameter of RED, $\alpha_{\text{red}}$, and the hyperparameter of the internal denoiser, $\alpha_{\text{denoiser}}$. 
Here, we propose a supervised/unsupervised parameter estimation method using deep algorithm unrolling.

In the unrolled version of Algorithm \ref{algorithm:proposed-red}, the following parameters are estimated: the set of RED regularization parameters for each layer (i.e., iteration) $\boldsymbol{\alpha}_{\text{red}} = \{\alpha_{\text{red}}^{(0)}, \alpha_{\text{red}}^{(1)}, \dots, \alpha_{\text{red}}^{(K)}\}$, and the set of hyperparameters in the internal denoiser for each layer $\boldsymbol{\alpha}_{\text{denoiser}} = \{\alpha_{\text{denoiser}}^{(0)}, \alpha_{\text{denoiser}}^{(1)}, \dots, \alpha_{\text{denoiser}}^{(K)}\}$, where $K$ is the number of iterations in the unrolled algorithm. 
This allows the algorithm to perform denoising with different strengths at different stages. 
As a result, the unrolled algorithm is expected to achieve not only fast convergence but also a high-quality solution compared to its original, non-unrolled counterpart. 

For supervised learning, trainable parameters are learned by using the ground-truth. 
On the other hand, for unsupervised learning, we perform Noise2Noise \cite{lehtinen2018Noise2Noise} in \eqref{eq:noise2noise}. 

\subsection{RED from Graph Filter Perspective}
\label{sec:proposed-graph-filter-perspective}
Here, we show the effectiveness of RED from the viewpoint of graph filters.
To understand their impact on denoising performance, we compare LR and RED.
Here, we assume LR is also solved by an iterative algorithm and 
the gradient update steps are compared.

If noise is sufficiently small ($\mathbf{x}^*-\mathbf{y} \simeq \mathbf{0}$), the gradient of the regularization term of LR is given by
\begin{equation}
    \begin{split}
        \nabla \left(\frac{1}{2}\|\mathbf{x}-\mathbf{y}\|^2_2 + \frac{\alpha_\text{lr}}{2}\mathbf{x}^\top\mathbf{L}\mathbf{x}\right) &= \mathbf{x}-\mathbf{y} + \alpha_\text{lr}\mathbf{L}\mathbf{x}\\
        & \sim \alpha_{\text{lr}} \mathbf{L}\mathbf{x}\\
        & = \mathbf{U}h_{\text{lr}}(\boldsymbol{\Lambda}, \alpha_\text{lr})\mathbf{U}^\top\mathbf{x},
    \end{split}
\end{equation}
where we use eigendecomposition of graph Laplacian and $h_{\text{lr}}(\boldsymbol{\Lambda}, \alpha_\text{lr})=\alpha_{\text{lr}}\boldsymbol{\Lambda}$.
Therefore, the gradient of LR as a graph filter can be represented as a graph high-pass filter.

We then consider the gradient of the regularization term of RED.
For simplicity, we use LR as an internal denoiser.
Its gradient is given by
\begin{equation}
\begin{split}
    \nabla \left(\frac{\alpha_{\text{red}}}{2}\mathbf{x}^\top(\mathbf{x}-\mathcal{D}_{\text{lr}}(\mathbf{x}))\right) &= \alpha_{\text{red}}\left(\mathbf{x}-\mathcal{D}_{\text{lr}}(\mathbf{x})\right)\\
    &= \alpha_\text{red}\left(\mathbf{x}-(\mathbf{I}+\alpha_\text{lr}\mathbf{L})^{-1}\mathbf{x}\right)\\
    &= \mathbf{U}h_\text{red}(\boldsymbol{\Lambda}, \alpha_\text{red})\mathbf{U}^\top\mathbf{x},
\end{split}
\end{equation}
where $h_\text{red}(\boldsymbol{\Lambda}, \alpha_\text{red})$ is the graph filter of RED and it can be represented as
\begin{equation}
\begin{split}
	h_{\text{red}}(\boldsymbol{\Lambda}, \alpha_\text{red}) &= \alpha_{\text{red}}\left( \mathbf{I}-( \mathbf{I}+ h_{\text{lr}}(\boldsymbol{\Lambda}, \alpha_\text{lr}))^{-1}\right)\\
&=\alpha_{\text{red}}\left(\mathbf{I}+ h_{\text{lr}}(\boldsymbol{\Lambda}, \alpha_\text{lr})\right)^{-1}\cdot h_{\text{lr}}(\boldsymbol{\Lambda}, \alpha_\text{lr}).
\end{split}
\end{equation}
In other words, $h_{\text{red}}(\boldsymbol{\Lambda}, \alpha_\text{red})$ can be regarded as a frequency-dependent scaled version of $h_{\text{lr}}(\boldsymbol{\Lambda}, \alpha_\text{lr})$.

Graph frequency responses of $h_\text{lr}(\boldsymbol{\Lambda}, \alpha_\text{lr})$ and $h_\text{red}(\boldsymbol{\Lambda}, \alpha_\text{red})$ are compared in Fig. \ref{fig:LRvsRED}.
It is clear that both regularizations work as graph high-pass filters, however, their spectral responses are different.
$h_\text{red}(\boldsymbol{\Lambda}, \alpha_\text{red})$ strongly attenuates high-frequency components in contrast to $h_\text{lr}(\boldsymbol{\Lambda}, \alpha_\text{lr})$: This results in that iterations in RED in \eqref{eq:RED-grad} \textit{preserve} high graph frequency components and eliminate oversmoothing. 

\section{Experiments}
\label{sec:experiments}
In this section, we demonstrate the effectiveness of the proposed method through graph signal denoising for synthetic and real-world data\footnote{Official implementation is available at \url{https://github.com/kojima-msp/GRED_DAU}.}.

\subsection{Setup}
\label{sec:experiments-setup}

\subsubsection{Datasets} \label{subsubsec:datasets}
We use following synthetic and real-world datasets for the experiments.

\noindent
\textbf{Synthetic Dataset: }
We use bandlimited graph signals as a synthetic dataset.
First, we randomly generate $N=100$ nodes in $[0,100] \times [0,100]$.
A $k$-nearest neighbor ($k$NN) graph ($k=5$) is constructed for the nodes where edge weights between nodes are calculated as \eqref{eq:calc-W} and normalized to lie within the range $[0, 1]$.
The original bandlimited graph signal $\mathbf{x}^*$ is then generated as
\begin{equation}
    \mathbf{x}^* = \mathbf{U}_{N_\text{band}} \mathbf{d},
\end{equation}
where $N_{\text{band}} \le N$ denotes the bandwidth, $\mathbf{U}_{N_\text{band}} \in \mathbb{R}^{N \times N_\text{band}}$ contains the first $N_\text{band}$ eigenvectors of $\mathbf{U}$, and $\mathbf{d} \in \mathbb{R}^{N_{\text{band}}}$ is defined as
\begin{equation}
    d_k = \sin\left( \frac{k\pi}{N_{\text{band}}} \right) + C,\quad 
    k=1, \dots, N_\text{band}
\end{equation}
with a constant value $C$. 
In this experiment, the original graph signal was generated with $N_{\text{band}} = 3$ and $C = 2$, and scaled to $[0, 100]$.

We use 10 signals as training data and 5 signals as test data. 
All graph signals are corrupted by additive white Gaussian noise $\mathcal{N}(0,\sigma^2)$ with $\sigma = \{10, 15, 20, 25, 30\}$.

\noindent
\textbf{Real-world Dataset: }
As a real-world dataset, we use ModelNet10 \cite{zhirongwu20153d}.
ModelNet10 is a 3-D polygon mesh dataset consisting of about 5000 CAD data in 10 classes.

In this experiment, we randomly select 10 objects as training data and 10 objects as test data.
Since the number of meshes is small, we first oversample them using a mesh subdivision method \cite{loop1987smooth} to treat them as dense 3-D point clouds. 
After that, we downsample them using farthest point sampling (FPS) \cite{eldar1994Farthest} so that the maximum number of nodes in each objects is $N=500$.
As in the case of synthetic data, we first construct a graph. Each node is connected by an edge based on the $k$NN with $k=5$, and the edges are weighted by \eqref{eq:calc-W}. 
The graph is then corrupted by additive white Gaussian noise $\mathcal{N}(0,\sigma^2)$ with $\sigma = \{10, 15, 20, 25, 30\}$. 

\begin{table*}[tp]
\footnotesize
\centering
\caption{Overview of comparison and proposed methods. ``S:'' and ``U:'' denote ``Supervised'' and ``Unsupervised,'' respectively.}
\label{table:methods}
\begin{tabular}{l|l|l}
\bhline{1.1pt}
Method type & \multicolumn{1}{c|}{Name} & \makecell{\# of \\Parameters} \\
\hline\hline
\multirow{4}{*}{S: Model-based} & Laplacian Regularization (LR) \cite{pang2017Graph} & 1 \\
 & Plug-and-play ADMM (PnP) \cite{yazaki2019Interpolation} & 2 \\
 & RED (LR) & 2 \\
 & RED (PnP) & 3 \\ \hline
\multirow{2}{*}{S: DAU-based} & RED (LR-DAU) & 22 \\
 & RED (PnP-DAU) & 33 \\ \hline\hline
\multirow{2}{*}{U: Data-driven} & Graph Attention Network (GAT) \cite{velickovic2018Graph} & 8501 \\
 & Untrained Graph Neural Networks (UGNN) \cite{rey2022Untrained} & 7550 \\ \hline
U: DAU-based & RED (LR-Unsupervised) & 22 \\
\bhline{1.1pt}
\end{tabular}
\end{table*}

\subsubsection{Comparison Methods} \label{exp:comparison-methods}
Methods to be compared are summarized in Table \ref{table:methods}.
They include both of model-based and data-driven methods.
Hyperparameters in the model-based methods are determined from training data using Bayesian optimization \cite{optuna_2019}. 
The network architecture of GAT is the same as \cite{rey2022Untrained}. 
The number of iterations $K$ of the proposed method and PnP-ADMM is set to $K=10$, considering the trade-off between denoising performance and computational cost.
The internal graph denoiser $\mathcal{D}_{\text{graph}}$ for PnP-ADMM is LR.
The proposed method uses two denoisers, LR or PnP-ADMM, for regularization: Their abbreviations are RED (LR) and RED (PnP), respectively. 
In the proposed supervised and unsupervised methods, parameters are trained using Adam \cite{DBLP:journals/corr/KingmaB14} with the learning rate $0.01$. 
In the proposed unsupervised method, $\sigma_{\text{N2N}}$ in \eqref{eq:n2n} is randomly selected from $[0, 0.4 \cdot \max(\mathbf{|x|})]$.

\subsection{Results: Synthetic Dataset}
\noindent
\textbf{RMSE: }
Table \ref{table:result-syn} shows the average of root mean square errors (RMSEs) for the synthetic dataset across various $\sigma$.

In the supervised setting, our model-based proposed methods, RED (LR) and RED (PnP), consistently achieve lower RMSEs than the conventional LR and PnP methods across all noise levels. 
This generally indicates the effectiveness of the RED algorithm for graph signal denoising beyond image processing.

The two model-based RED methods, RED (LR) and RED (PnP) show comparable RMSEs.
This could be the simple structure of the signals: Later, we show RED (PnP) is slightly better than RED (LR) for more complex-structured data in Section \ref{sec:results-real}.

The DAU-incorporated RED methods, RED (LR-DAU) and RED (PnP-DAU), significantly outperform their respective model-based counterparts. 
This demonstrates the power of DAU in learning iteration-specific parameters, leading to further RMSE improvements.
Notably, RED (PnP-DAU) consistently yields lower RMSEs than RED (LR-DAU) across almost all noise levels. 
A primary reason for this superiority is that iteration-specific parameters in DAU help the PnP algorithm converge faster. 

In the unsupervised setting, RED (LR-Unsupervised) demonstrates superior denoising performance compared to existing data-driven GNNs like GAT and UGNN. 
This even with significantly fewer parameters.

\begin{table*}[tp]
\centering
\small
\caption{RMSEs using synthetic dataset. The best results for supervised and unsupervised methods are shown in \textbf{bold}.}
\label{table:result-syn}
\begin{tabular}{l|ccccc}
\bhline{1.1pt}
 & $\sigma=10$ & 15 & 20 & 25 & 30 \\
\hline\hline
Observed & 10.09 & 14.74 & 19.48 & 25.23 & 32.00 \\ \hline
LR & 3.75 & 5.02 & 5.31 & 6.78 & 9.57 \\
PnP & 3.81 & 5.08 & 5.40 & 6.88 & 9.69 \\
RED (LR) & 3.67 & 4.85 & 4.89 & 6.50 & 9.28 \\
RED (PnP) & 3.61 & 4.83 & 4.92 & 6.51 & 9.29 \\ \hline
RED (LR-DAU) & 3.41 & 4.51 & 4.58 & 5.33 & \textbf{7.35} \\
RED (PnP-DAU) & \textbf{2.60} & \textbf{3.02} & \textbf{4.51} & \textbf{4.18} & 7.67 \\ \hline
GAT & 6.51 & 10.31 & \textbf{10.73} & 14.56 & 27.14 \\
UGNN & 9.15 & 13.40 & 19.47 & 22.63 & 29.25 \\ \hline
RED (LR-Unsupervised) & \textbf{5.73} & \textbf{7.91} & 11.18 & \textbf{11.55} & \textbf{19.52} \\ \bhline{1.1pt}
\end{tabular}
\end{table*}

\noindent
\textbf{Visualization: }
Denoised signals are visualized in Fig. \ref{fig:result_syn_20}. 
These results validate the quantitative RMSE improvements.

While existing model-based methods, such as LR and PnP, show some denoising effects, areas with larger errors are still observable compared to the model-based versions of RED, RED (LR) and RED (PnP).
The DAU-based proposed methods, RED (LR-DAU) and RED (PnP-DAU), further reduce the errors from their respective model-based counterparts, suggesting the effectiveness of parameter learning through DAU.

In unsupervised learning, existing GNN-based methods, GAT and UGNN, represent large errors compared to RED (LR-Unsupervised).
In contrast, RED (LR-Unsupervised) achieves significantly smaller errors than GAT and UGNN, indicating good denoising performance.

\begin{figure*}
  \centering
  \begin{minipage}[t]{0.30\linewidth}
    \centering
    \includegraphics[width=3.5cm]{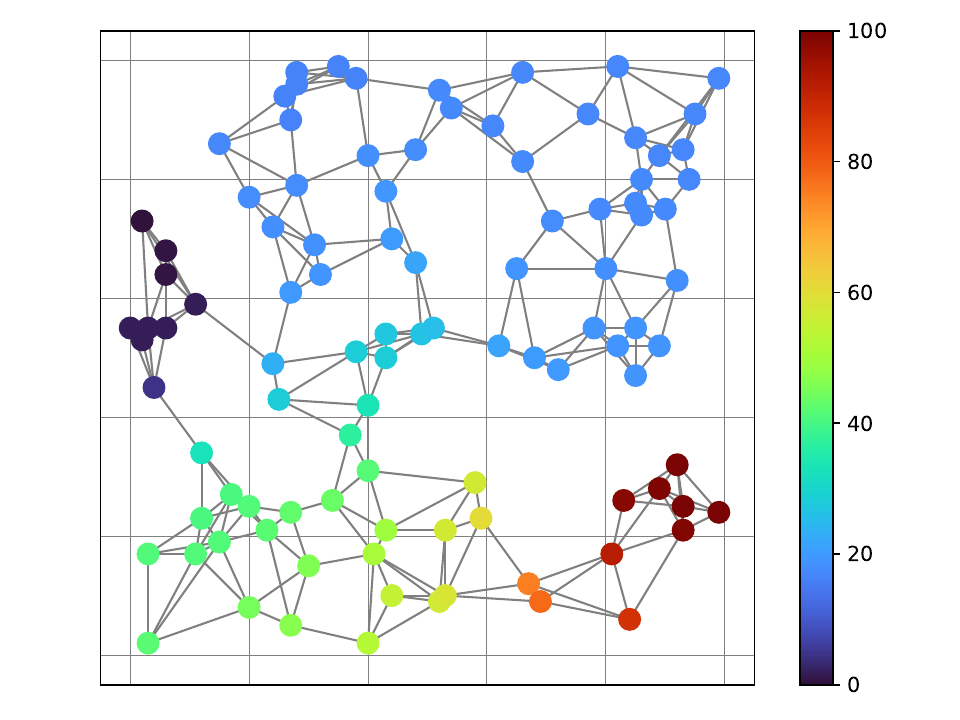}
    \subcaption{Original signal. }
  \end{minipage}
  \begin{minipage}[t]{0.30\linewidth}
    \centering
    \includegraphics[width=3.5cm]{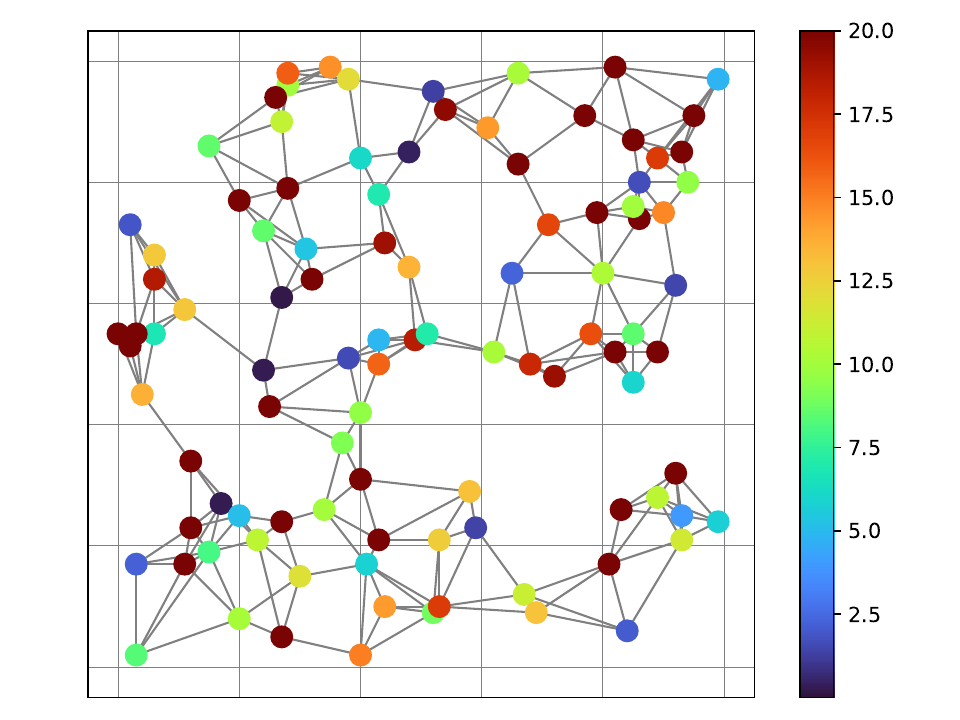}
    \subcaption{Observed signal \\(RMSE:20.03). }
  \end{minipage}
  \begin{minipage}[t]{0.30\linewidth}
    \centering
    \includegraphics[width=3.5cm]{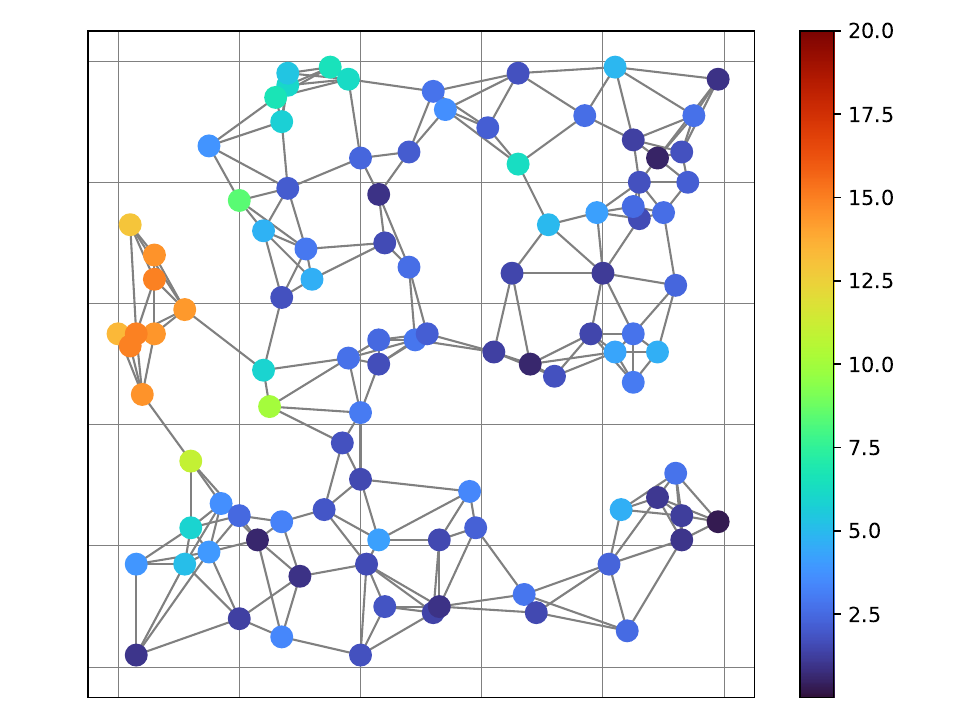}
    \subcaption{LR (5.51). }
  \end{minipage}
  \begin{minipage}[t]{0.30\linewidth}
    \centering
    \includegraphics[width=3.5cm]{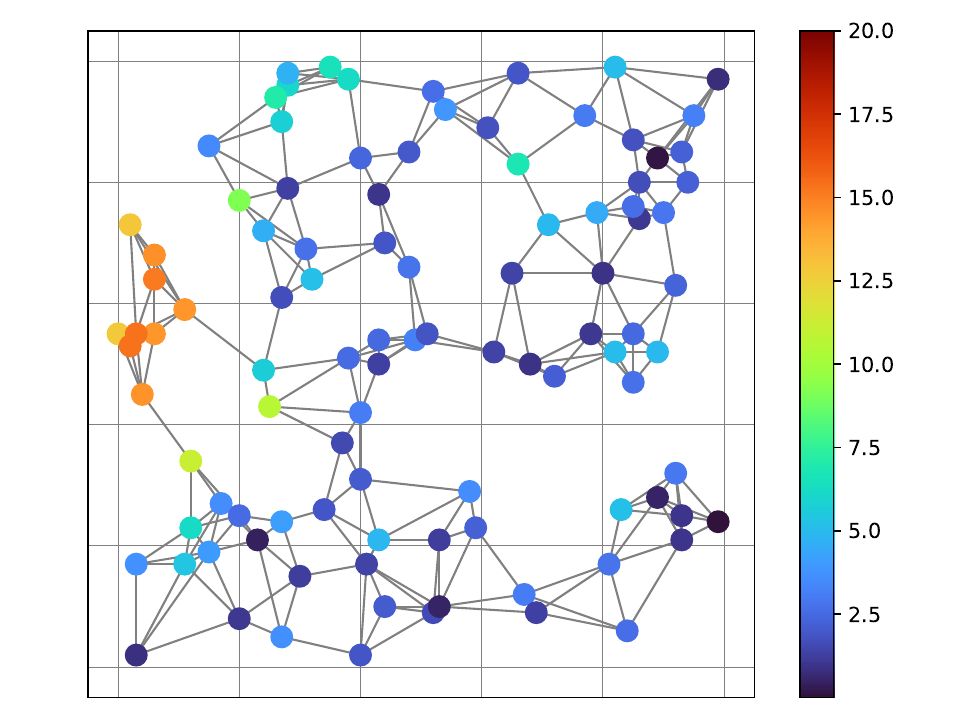}
    \subcaption{PnP (5.60). }
  \end{minipage}
  \begin{minipage}[t]{0.30\linewidth}
    \centering
    \includegraphics[width=3.5cm]{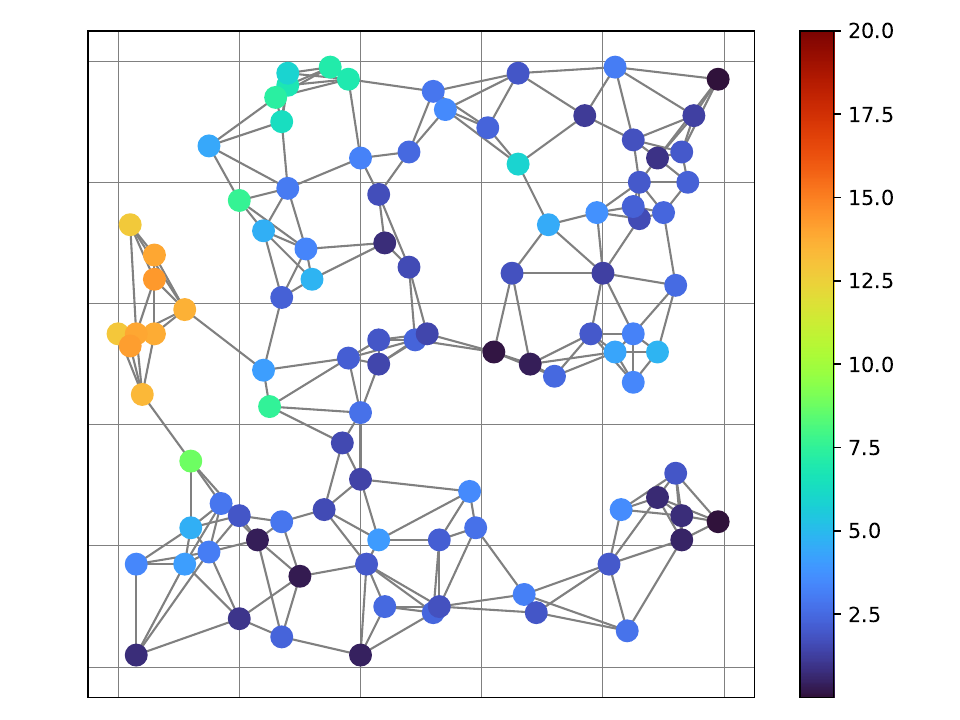}
    \subcaption{RED (LR) (5.21). }
  \end{minipage}
  \begin{minipage}[t]{0.30\linewidth}
    \centering
    \includegraphics[width=3.5cm]{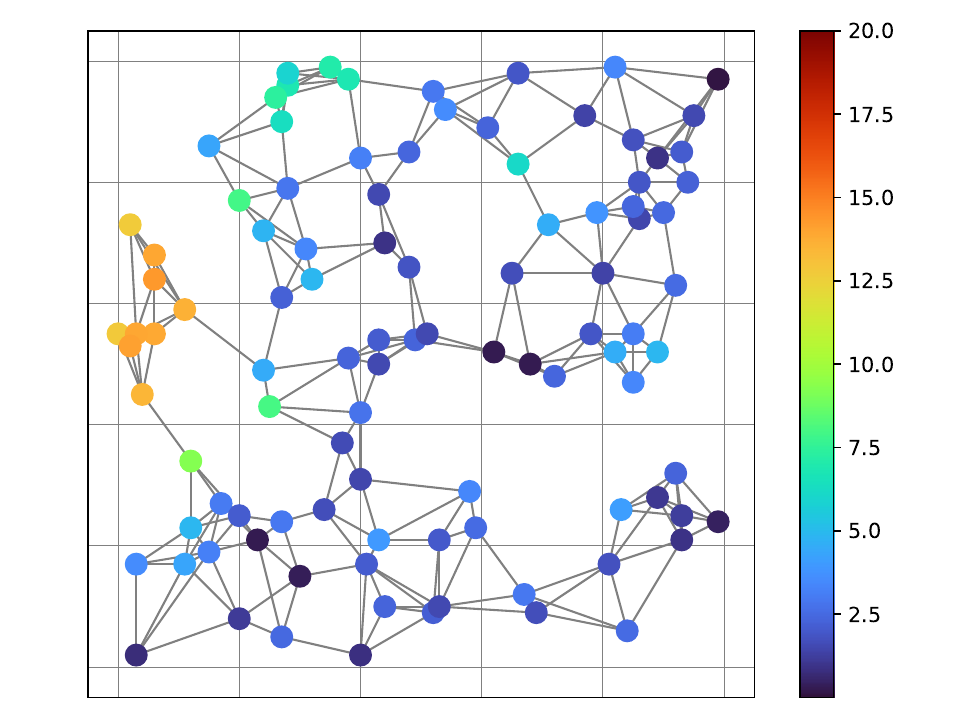}
    \subcaption{RED (PnP) (5.25). }
  \end{minipage}
  \begin{minipage}[t]{0.30\linewidth}
    \centering
    \includegraphics[width=3.5cm]{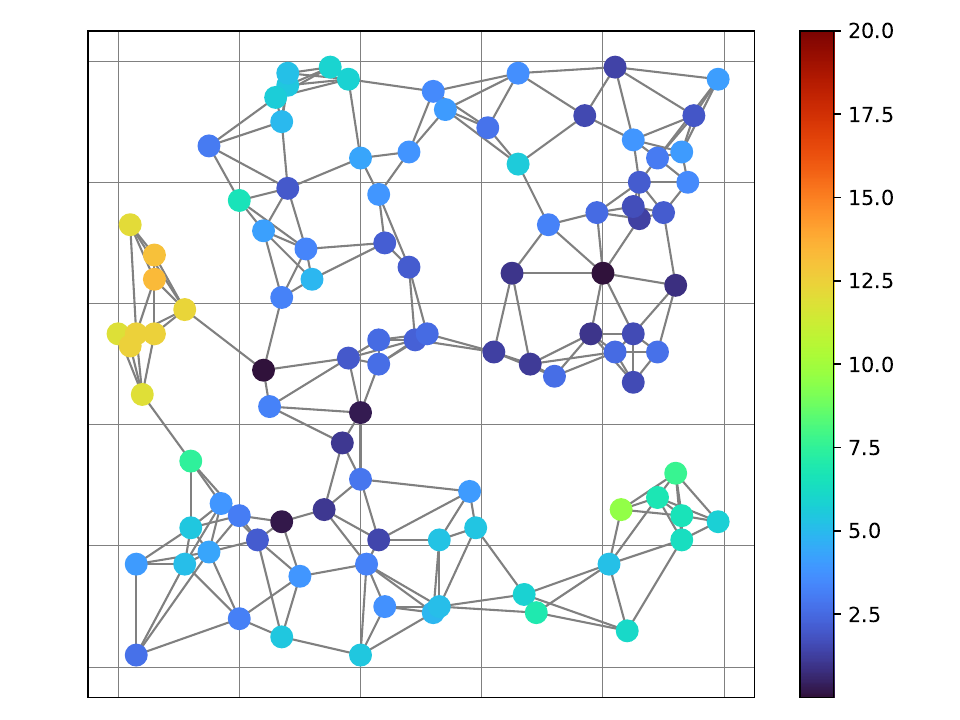}
    \subcaption{RED (LR-DAU) (5.38). }
  \end{minipage}
  \begin{minipage}[t]{0.30\linewidth}
    \centering
    \includegraphics[width=3.5cm]{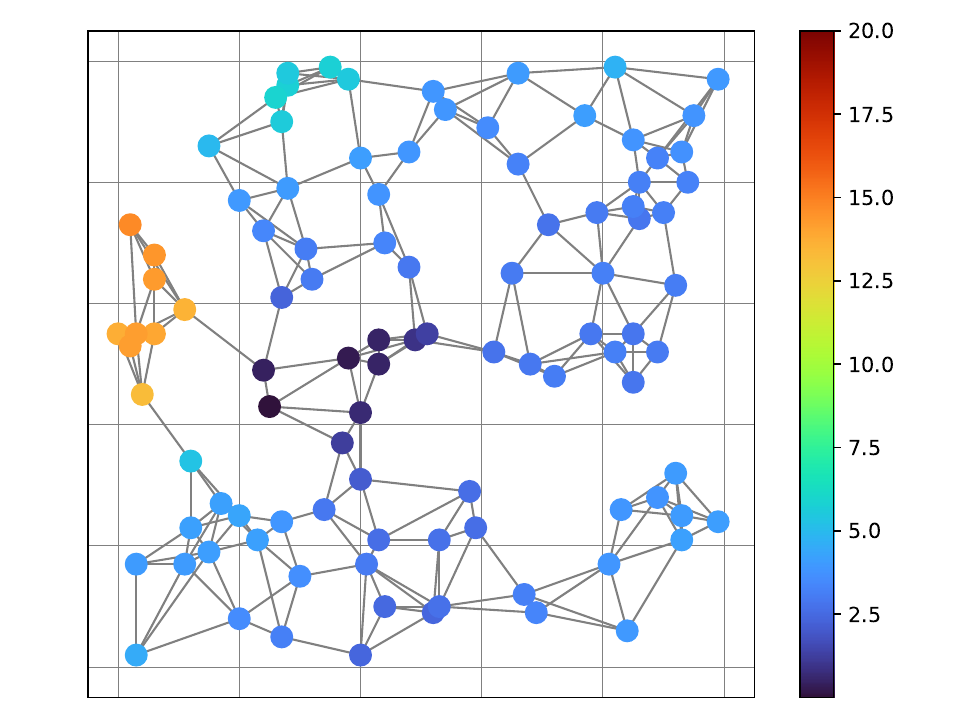}
    \subcaption{RED (PnP-DAU) \\(5.39). }
  \end{minipage}
  \begin{minipage}[t]{0.30\linewidth}
    \centering
    \includegraphics[width=3.5cm]{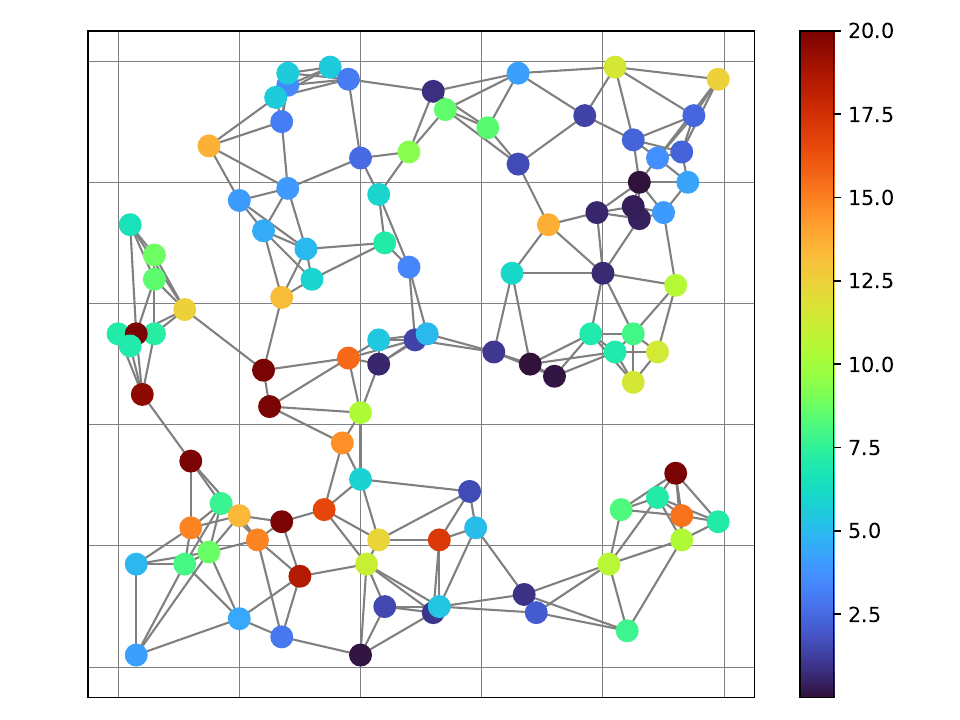}
    \subcaption{GAT (9.93). }
  \end{minipage}
  \begin{minipage}[t]{0.30\linewidth}
    \centering
    \includegraphics[width=3.5cm]{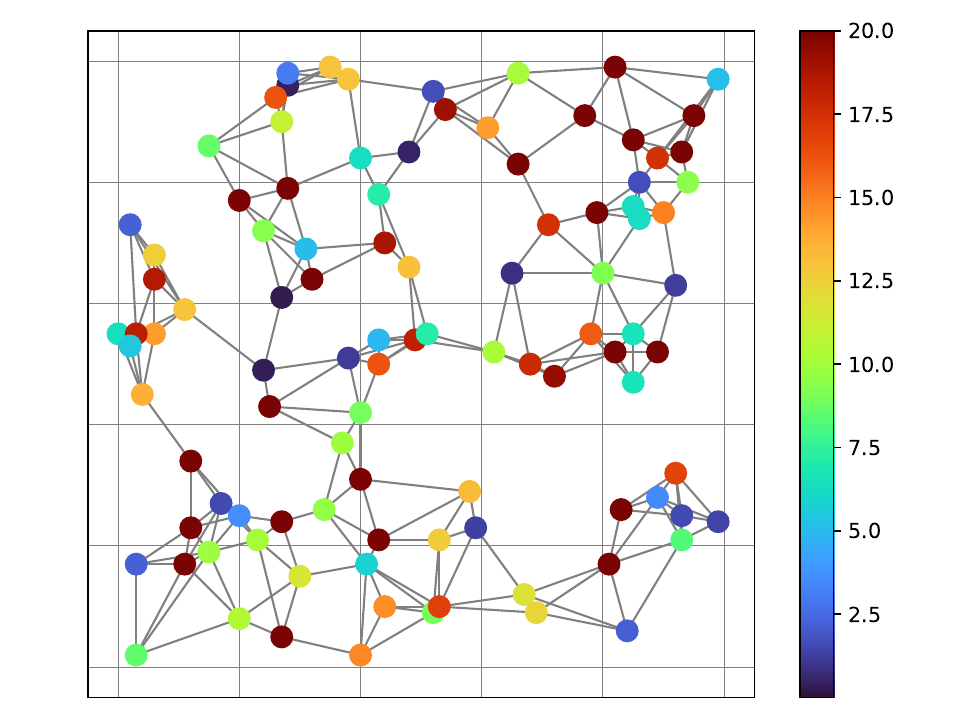}
    \subcaption{UGNN (19.12). }
  \end{minipage}
  \begin{minipage}[t]{0.30\linewidth}
    \centering
    \includegraphics[width=3.5cm]{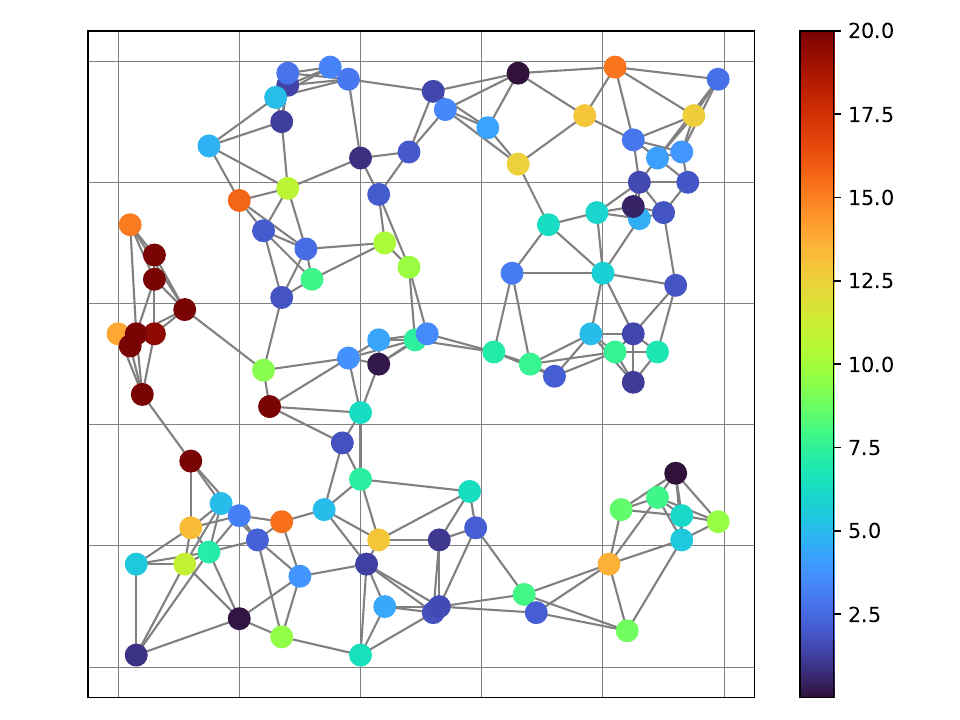}
    \subcaption{RED \\(LR-Unsupervised) (9.23). }
  \end{minipage}
  \begin{minipage}[t]{0.30\linewidth}
  \centering
    \phantom{\includegraphics[width=3.5cm]{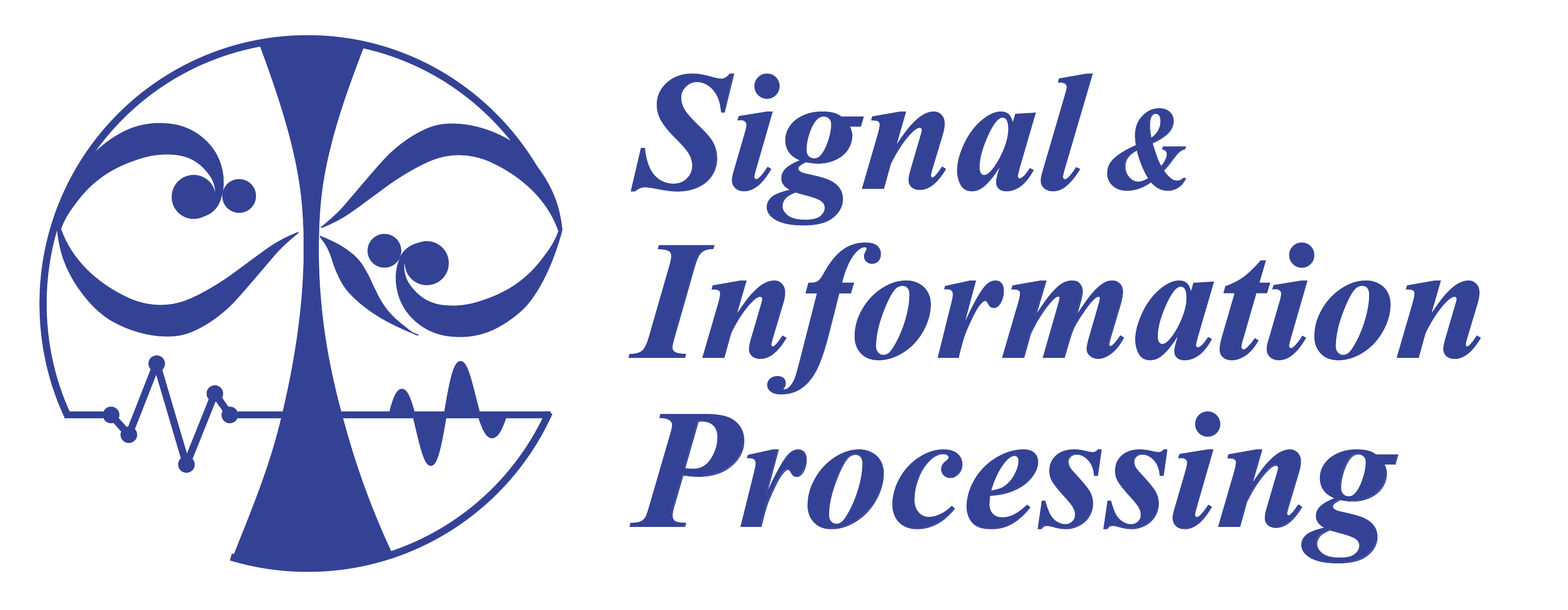}}
  \end{minipage}
  \caption{Visualization results of the proposed and existing methods using synthetic dataset ($\sigma=20$). (b)--(k) show the error $\tilde{\mathbf{x}} - \mathbf{x}$ between the restored signal $\tilde{\mathbf{x}}$ and the original signal $\mathbf{x}$. The values in parentheses indicate the RMSE for this specific sample. Note that these values differ from the average results reported in Table \ref{table:result-syn}.}
  \label{fig:result_syn_20}
\end{figure*}

\subsection{Results: Real-world Dataset}
\label{sec:results-real}
\noindent
\textbf{RMSE: }
Table \ref{table:result-real} presents the RMSE results on the real-world dataset.

In the supervised model-based methods, the proposed methods, RED (LR) and RED (PnP), outperform their baselines, LR and PnP: It is consistent with the synthetic data results. 
On the real-world data, RED (PnP) outperforms RED (LR) at high noise levels ($\sigma=20, 25, 30$). 
This suggests that for the complex geometric structures of the dataset, RED (PnP) may be more effective.
Note that, for a fair comparison, we use the fixed $K$ both for RED (LR) and RED (PnP) while we usually need a large $K$ for PnP-based methods: This may result in further performance improvements.

For the DAU-based supervised methods, we again observe a significant performance improvement over the non-DAU counterparts, which is consistent with the synthetic data. 
Interestingly, while RED (PnP-DAU) is consistently superior for synthetic data, RED (LR, DAU) shows better performance on the real-world data at almost all noise levels. 

In the unsupervised setting, consistent with the synthetic data results, RED (LR-Unsupervised) demonstrates superior performance to the existing GNN-based methods.

\begin{table*}[tp]
\centering
\small
\caption{RMSEs using real-world dataset. The best results for supervised and unsupervised methods are shown in \textbf{bold}.}
\label{table:result-real}

\begin{tabular}{l|ccccc}
\bhline{1.1pt}
 & $\sigma=10$ & 15 & 20 & 25 & 30 \\
\hline\hline
Observed & 10.00 & 15.05 & 19.84 & 24.90 & 29.82 \\ \hline
LR & 3.50 & 4.08 & 4.89 & 5.63 & 6.51 \\
PnP & 3.46 & 4.08 & 4.93 & 5.71 & 6.62 \\
RED (LR) & 3.44 & 4.01 & 4.81 & 5.55 & 6.45 \\
RED (PnP) & 3.58 & 4.07 & 4.80 & 5.44 & 6.30 \\ \hline
RED (LR, DAU) & \textbf{3.18} & 3.94 & \textbf{4.43} & \textbf{4.98} & \textbf{5.64} \\
RED (PnP-DAU) & 3.50 & \textbf{3.87} & 4.47 & 5.37 & 6.12 \\ \hline
GAT & 6.78 & 7.26 & 9.88 & 11.36 & 16.51 \\
UGNN & \textbf{4.67} & 6.73 & 8.91 & 11.06 & 13.31 \\ \hline
RED (LR-Unsupervised) & 4.86 & \textbf{5.53} & \textbf{7.12} & \textbf{9.33} & \textbf{10.97} \\ \bhline{1.1pt}
\end{tabular}
\end{table*}

\noindent
\textbf{Visualization: }
Figure \ref{fig:result_3dpc_20} illustrates visual examples of denoised \textit{Chair} in the ModelNet10 at $\sigma=20$.

As shown in Figs. \ref{fig:result_3dpc_20}(e) and (f), RED (LR) and RED (PnP) demonstrate clear improvements over the baseline LR and PnP methods. 
RED (LR, DAU) and RED (PnP-DAU) achieve a significant reduction in RMSE as seen in Figs. \ref{fig:result_3dpc_20}(g) and (h). 
These methods effectively reduce oversmoothing artifacts, which appears as a loss of sharp features in the baseline methods. 

In the unsupervised setting, denoised signals by the proposed methods present substantially clearer object details than those yielded by alternative methods as shown in Figs. \ref{fig:result_3dpc_20}(i)--(k). 
This suggests that the unsupervised training framework using Noise2Noise successfully learns parameters taking into account unknown strengths of noise without requiring ground-truth data.

Overall, these visualization results are consistent with the numerical performances: They validate the effectiveness of the proposed RED-based graph signal denoising approaches.

\begin{figure*}[tp]
  \centering
  \begin{minipage}[t]{0.3\linewidth}
    \centering
    \includegraphics[width=3.5cm]{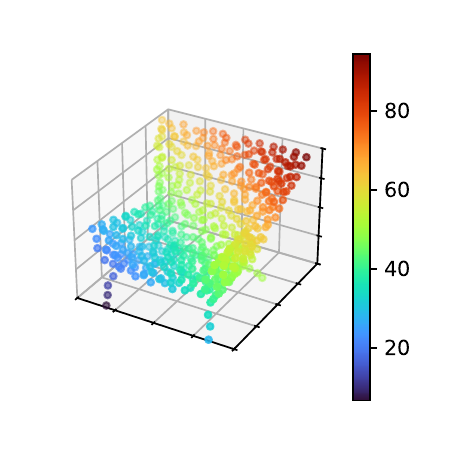}
    \subcaption{Original signal. }
  \end{minipage}
  \begin{minipage}[t]{0.3\linewidth}
    \centering
    \includegraphics[width=3.5cm]{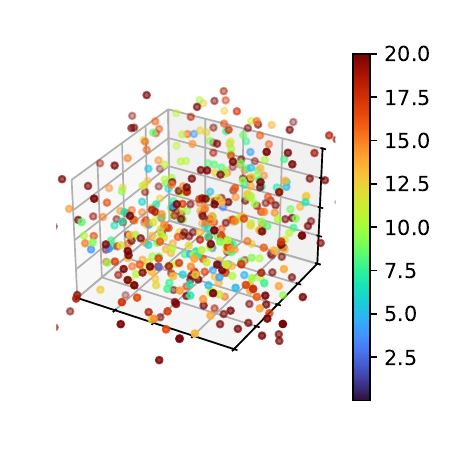}
    \subcaption{Observed signal (RMSE:19.88). }
  \end{minipage}
  \begin{minipage}[t]{0.3\linewidth}
    \centering
    \includegraphics[width=3.5cm]{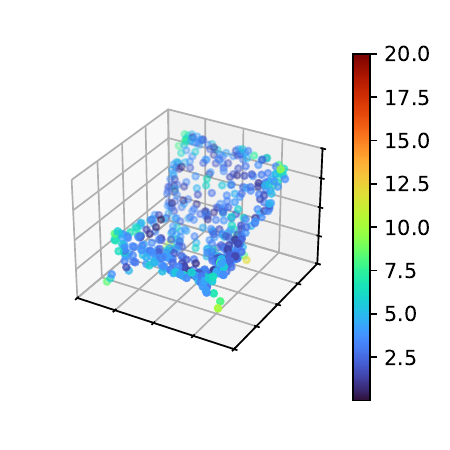}
    \subcaption{LR (4.84). }
  \end{minipage}
  \begin{minipage}[t]{0.3\linewidth}
    \centering
    \includegraphics[width=3.5cm]{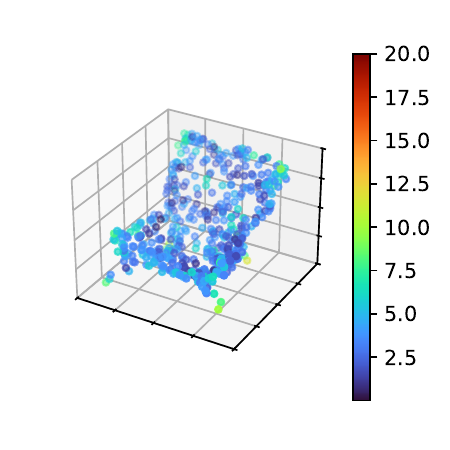}
    \subcaption{PnP (4.82). }
  \end{minipage}
  \begin{minipage}[t]{0.3\linewidth}
    \centering
    \includegraphics[width=3.5cm]{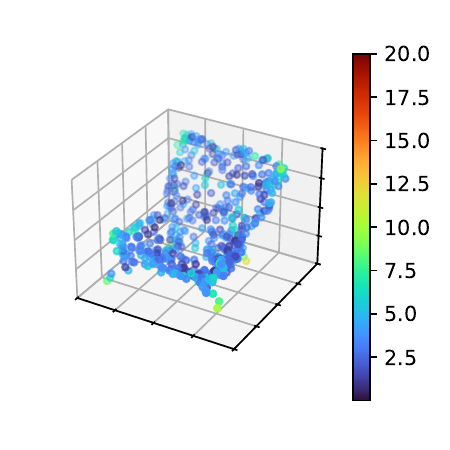}
    \subcaption{RED (LR) (4.63). }
  \end{minipage}
  \begin{minipage}[t]{0.3\linewidth}
    \centering
    \includegraphics[width=3.5cm]{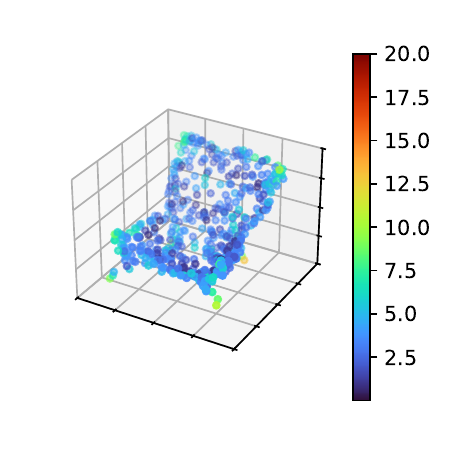}
    \subcaption{RED (PnP) (4.81). }
  \end{minipage}
  \begin{minipage}[t]{0.3\linewidth}
    \centering
    \includegraphics[width=3.5cm]{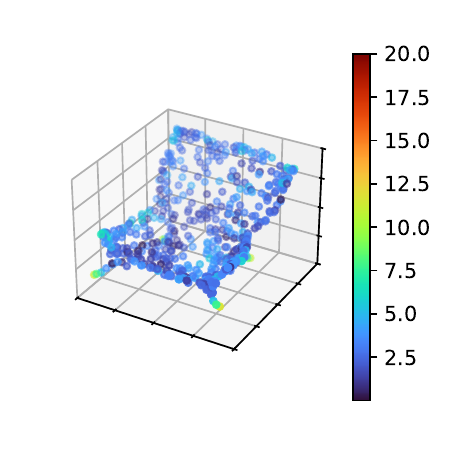}
    \subcaption{RED (LR, DAU) (3.94). }
  \end{minipage}
  \begin{minipage}[t]{0.3\linewidth}
    \centering
    \includegraphics[width=3.5cm]{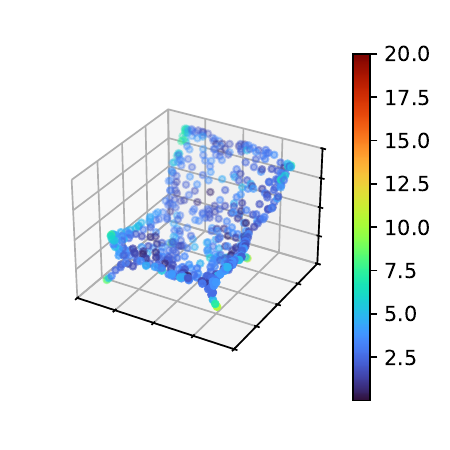}
    \subcaption{RED (PnP-DAU) \\(4.10). }
  \end{minipage}
  \begin{minipage}[t]{0.3\linewidth}
    \centering
    \includegraphics[width=3.5cm]{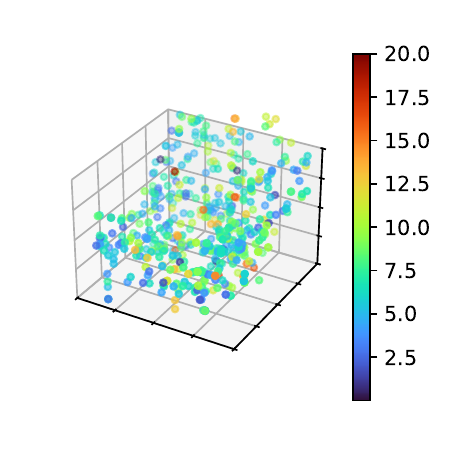}
    \subcaption{GAT (8.84). }
  \end{minipage}
  \begin{minipage}[t]{0.3\linewidth}
    \centering
    \includegraphics[width=3.5cm]{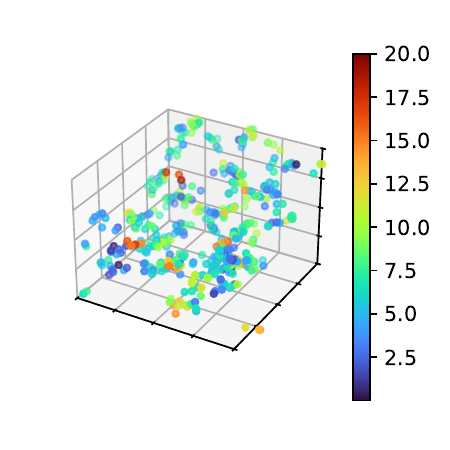}
    \subcaption{UGNN (8.66). }
  \end{minipage}
  \begin{minipage}[t]{0.3\linewidth}
    \centering
    \includegraphics[width=3.5cm]{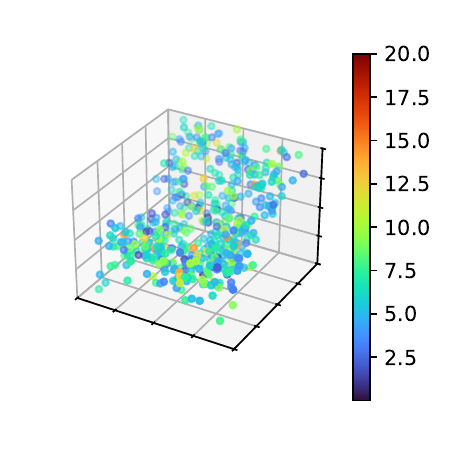}
    \subcaption{RED \\(LR-Unsupervised) (7.67). }
  \end{minipage}
  \begin{minipage}[t]{0.30\linewidth}
    \centering
    \phantom{\includegraphics[width=3.5cm]{img/SIP_logo.pdf}}
  \end{minipage}
  \caption{Visualization results of the proposed and existing methods using real-world dataset ($\sigma=20$). The colors of (b)--(k) show the error $\tilde{\mathbf{x}} - \mathbf{x}$ between the restored signal $\tilde{\mathbf{x}}$ and the original signal $\mathbf{x}$. The values in parentheses indicate the RMSE for this specific sample. Note that these values differ from the average results reported in Table \ref{table:result-real}.}
  \label{fig:result_3dpc_20}
\end{figure*}

\subsection{Computation Time} \label{exp:computation-time}
Table \ref{table:computation-time} shows the average computation time for the model- and RED-based methods \footnote{All experiments were performed on a workstation with an AMD Ryzen 7 9700X, NVIDIA GeForce RTX 3090, and 64 GB of memory.}. 
The proposed methods exhibit longer runtimes compared to their baseline counterparts (LR and PnP) because they involve iterative calls to these denoisers within the RED algorithm. 
The increase is particularly pronounced for RED (PnP), which has a hierarchical structure: PnP iterations are nested within RED iterations, and LR is computed within PnP.
We consider this increased computational cost as a trade-off for achieving higher denoising performance. 
As shown in Tables \ref{table:result-syn} and \ref{table:result-real}, the proposed methods significantly improve the RMSE that justify the additional runtime.

A comparison across datasets also reveals that, the execution time tends to increase for almost all methods as the number of nodes increases. 
This is mainly attributed to the matrix inversion of the graph Laplacian. 
However, this bottleneck can be mitigated by employing iterative algorithms, such as the conjugate gradient method, instead of computing the LR denoiser in a closed form. 

\begin{table}[]
    \centering
    \small
    \caption{Average computation time of each method [ms].}
    \begin{tabular}{l|cc}\bhline{1.1pt}
        Method & Synthetic ($N=100$) & 3-D point cloud ($N=500$) \\ \hline\hline
        LR & 5.23 & 4.97 \\
        PnP & 5.03 & 5.65\\ \hline
        RED (LR) & 6.76 & 19.96\\
        RED (PnP) & 27.01 & 284.46 \\
        RED (LR, DAU) & 7.32 & 18.35\\
        RED (PnP, DAU) & 27.07 & 284.58\\ \bhline{1.1pt}
    \end{tabular}
    \label{table:computation-time}
\end{table}

\section{Conclusion}
\label{sec:conclusion}
In this paper, we propose an interpretable denoising method for graph signals using regularization by denoising (RED). 
We extend the applicability of RED beyond image processing, demonstrating its effectiveness in a wide range of graph signal denoisers. 
Additionally, we introduce both supervised and unsupervised training approaches based on deep algorithm unrolling and Noise2Noise. 
Through a graph filter perspective, we reveal the advantages of RED in the context of graph signal denoising. 
Experimental results on synthetic and real-world data show that our method achieves better performance than existing graph signal denoising techniques, highlighting its potential for future downstream applications.

\printbibliography

\vfill\pagebreak

\end{document}